\documentclass[10pt,journal,compsoc]{IEEEtran}

\let\labelindent\relax

\usepackage{amsmath,amssymb,paralist,subfigure,cite,graphicx,color,amsbsy,amsthm}
\usepackage{stmaryrd,mathrsfs,url}
\usepackage[table]{xcolor}
\usepackage{cuted,mathtools,lipsum}
\usepackage[ruled,vlined,linesnumbered]{algorithm2e}

\usepackage{enumerate}   
\usepackage{pifont}
\usepackage[utf8]{inputenc} % allow utf-8 input
\usepackage{hyperref}       % hyperlinks
\hypersetup{
    colorlinks=true,
    linkcolor=cyan,
    filecolor=mnodea,
    urlcolor=cyan,
    citecolor=lime,
}
\usepackage{url}            % simple URL typesetting
\usepackage{booktabs}       % professional-quality tables
\usepackage{amsfonts}       % blackboard math
\usepackage{amsthm}     % new theorems, definitions
\usepackage{nicefrac}       % compact symbols for 1/2, etc.
\usepackage{microtype}      % microtypography
\usepackage{lipsum}
\usepackage{graphicx}
\usepackage{dsfont}
\usepackage{enumitem}

\usepackage{multicol}
\usepackage{breqn}
\usepackage{stfloats}
\newtheorem{lem}{Lemma}
\newtheorem{ass}{Assumption}
\newtheorem{theorem}{Theorem}
\newtheorem{defn}{Definition}
\newtheorem{rem}{Remark}

\def\mb{\mathbf}

\def\mc{\mathcal}

\begin{document}
\title{\bf On the Design of Resilient Distributed \\ Single Time-Scale Estimators: A Graph-Theoretic Approach}
\author{Mohammadreza  Doostmohammadian,  Mohammad Pirani 
	\IEEEcompsocitemizethanks{\IEEEcompsocthanksitem M.~Doostmohammadian is with the Mechatronics Department, Faculty of Mechanical Engineering, Semnan University, Semnan, Iran {\texttt{doost@semnan.ac.ir}}.
		\IEEEcompsocthanksitem M. Pirani is with the Department of Mechanical Engineering, University of Ottawa, Canada, email: {\texttt{mpirani@uottawa.ca}}.
	}
}

\markboth{IEEE Transaction on Network Science and Engineering}%
{Doostmohammadian \MakeLowercase{\textit{et al.}}: On the Design of Resilient Distributed Estimators }

\IEEEtitleabstractindextext{
	\begin{abstract}
	Distributed estimation in interconnected systems has gained increasing attention due to its relevance in diverse applications such as sensor networks, autonomous vehicles, and cloud computing. In real practice, the sensor network may suffer from communication and/or sensor failures. This might be due to cyber-attacks, faults, or environmental conditions. Distributed estimation resilient to such conditions is the topic of this paper. By representing the sensor network as a graph and exploiting its inherent structural properties, we introduce novel techniques that enhance the robustness of distributed estimators.  As compared to the literature, the proposed estimator (i) relaxes the network connectivity of most existing single time-scale estimators and (ii) reduces the communication load of the existing double time-scale estimators by avoiding the inner consensus loop. 
	On the other hand, the sensors might be subject to faults or attacks, resulting in biased measurements. Removing these sensor data may result in observability loss. Therefore, we propose resilient design on the definitions of $q$-node-connectivity and $q$-link-connectivity, which capture robust strong-connectivity under link or sensor node failure. By proper design of the sensor network, we prove Schur stability of the proposed distributed estimation protocol under failure of up to $q$ sensors or $q$ communication links. 
	\end{abstract}

\begin{IEEEkeywords}
	Distributed observability, graph theory, networked estimation, consensus, $q$-node-connected networks, $q$-link-connected networks
\end{IEEEkeywords}

}
\maketitle
\IEEEdisplaynontitleabstractindextext
\IEEEpeerreviewmaketitle

\IEEEraisesectionheading{\section{Introduction}\label{sec_intro}}
\IEEEPARstart{R}{esilient} design of sensor networks for monitoring purposes has recently gained great interest in control and signal processing literature \cite{sundaram_2021resilient,PIRANI2023111264,battilotti2018distributed,mohammadi2015distributed}. In the era of interconnected and distributed systems, the need for accurate and resilient estimators has become more important. From monitoring complex networks to managing critical infrastructure, distributed estimators play a key role in providing real-time insights and ensuring system stability. However, the design of such estimators faces significant challenges, including the vulnerability to node failures, communication bottlenecks, and adversarial attacks. Addressing these challenges requires a new perspective that combines advanced techniques from graph theory into the distributed estimation strategy. This paper explores a novel approach that uses the notion of graph theory to design resilient distributed estimators, providing a robust foundation for critical applications in various domains, including distributed target tracking \cite{tase,ennasr2016distributed,mohammadi2015distributed,ennasr2018distributed} and distributed fault detection and isolation (FDI) \cite{teixeira2014distributed,tcns20,hajshirmohamadi2019event,davoodi2013distributed,ferrari2011distributed}. The existing literature on distributed estimation are based on different approaches; namely, observable canonical decomposition method \cite{deng2023distributed}, secure distributed estimation method against data integrity attacks \cite{wu2021secure},  distributed consensus-based estimation under saturation constraint \cite{jin2023new}, edge-pinning-based secure distributed estimation \cite{liu2023fully}, and finite-time weighted least-squares estimation\cite{zhu2022optimal}. Most existing works in the literature are based on double time-scale estimation which requires many steps of consensus filter between two consecutive steps of system dynamics, see Fig.~\ref{fig_scale}. This implies a consensus loop in the algorithm which increases computation and communication load on sensors/agents and algorithm complexity, however, it may result in more noise reduction and improved accuracy of the estimation process. Overall, there is a trade-off between the communication/computation loads and the error performance of the algorithm. The other key assumption is on the system observability. Many works \cite{kar2013consensus,das2015distributed,Silm2020A,khan2010connectivity,chen2018internet} assume local observability of the system in the neighborhood of each sensor which requires access to more sensor measurements and increased measurement-sharing over the sensor network. This necessitates more network connectivity and, therefore, more communications among the sensor nodes.  

\begin{figure}[]
	\centering
	\includegraphics[width=3.5in]{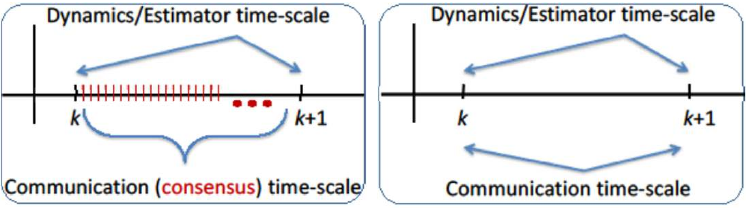}
	\caption{This figure illustrates two main strategies for distributed estimator/observer design in the literature. (Left) Double time-scale method with many steps of averaging (consensus) and communications between every two consecutive steps of system dynamics. The number of consensus steps is generally more than the diameter of the network. (Right) Single time-scale method (adopted in this work) with only one step of averaging (consensus) and communication, which imposes less communication and processing load on sensors/agents.}
	\label{fig_scale}
\end{figure}

In this work, we propose a new distributed estimation protocol that is single time-scale and has no inner consensus loop. In other words, the sensors perform only one step of  consensus (or averaging) between every two consecutive steps of system dynamics (in contrast to the double time-scale algorithms), see Fig.~\ref{fig_scale}. This estimator assumes no local observability, i.e., there is no assumption on the observability of the underlying system in the neighbourhood of any sensor. Further, in contrast to \cite{acc13}, the estimator only performs data-fusion on the a-priori estimates with no step of measurement sharing for full-rank systems. This significantly reduces the sensor network connectivity requirement.  
Our proposed distributed setup relaxes the network requirement to strong-connectivity and reaches convergence under such minimal connectivity. This allows for resilient estimator design by adding network redundancy. Note that the existing single time-scale estimation literature requires more than strong-connectivity. In the previous work \cite{acc13} a hub-based network is needed to satisfy distributed observability. Many other works
\cite{kar2013consensus,das2015distributed,Silm2020A,khan2010connectivity,chen2018internet} also assume observability in the neighbourhood of every sensor which mandates more than strong-connectivity. As we show in this paper strong-connectivity of the sensor network is the key to the resilient design and, therefore, many existing single time-scale distributed estimators in the literature cannot satisfy this property. On the other hand, double time-scale methods \cite{he2020secure,battilotti2021stability,olfati2007distributed} require an inner consensus loop to perform many steps of information-fusion between every two consecutive steps of system dynamics. In particular, the number of data-sharing steps needs to be more than the network diameter to satisfy system-observability requirements. This mandates a much faster communication rate than system dynamics which might be costly or infeasible in real-world applications. 

Our resilient design is based on the graph-theoretic notions of $q$-node-connectivity and $q$-link-connectivity. Following the strong-connectivity condition, by designing the network to be \textit{survivable} in the case of node or link failure, the resilient design guarantees that the condition for distributed observability holds. The case of node failure is more challenging. This is because we also need to ensure that centralized observability is preserved, which follows our relaxed observability setup and the concept of observational equivalence. By the use of observationally equivalent sensor measurements, we recover the loss of centralized observability while the survivable network design preserves connectivity requirements for distributed observability. There are many algorithms for survivable network design in the graph-theory literature, see \cite{galvez2021cycle,byrka2020breaching,cecchetto2021bridging,galvez2021approximation,gupta2011approximation} for example. Using these algorithms one can design the network to keep its connectivity in the case of node/link deletion or optimally add extra nodes/links to a given network to increase its $q$-connectivity.

We summarize our contributions in the following:
	\begin{itemize}
		\item The proposed single time-scale distributed estimator has no inner consensus loop and performs only one step of consensus filter between every two consecutive samples of the system dynamics. This significantly reduces the \textit{communication} and \textit{computational load} on sensors as compared to many existing literature on double time-scale estimators, where many steps of communication/consensus (more than the sensor network diameter) are required.
		\item The distributed estimator only requires global observability of the system to all sensors/agents, with \textit{no need of local observability} of the system to every individual agent. This significantly reduces the network connectivity and communication as compared to many existing literature on single time-scale estimators.
		\item The network connectivity requirement of the proposed distributed solution is more relaxed than the existing literature which allows for adding sufficient redundancy and survivable network design. This improves the \textit{resilience} of the estimator/observer to failure of certain number of agents and/or communication links.  
	\end{itemize}

The rest of the paper is organized as follows. Section~\ref{sec_fram} states the problem and preliminaries on system graph theory. Section~\ref{sec_res} provides the proposed resilient distributed estimator design and the stability analysis. Section~\ref{sec_sim} presents the simulation results and Section~\ref{sec_con} concludes the paper. 

\section{The Framework} \label{sec_fram}
\subsection{Statement of the Problem}
Given an LTI system, the problem is to design a distributed estimation network of sensors observing the system in a decentralized way while it is resilient to the removal (or failure) of up to $q$ sensor nodes. In other words, the sensor network and distributed estimation protocol need to be designed such that after removing up to $q$ sensor nodes the rest of the sensor network can infer/estimate the entire state of the system. The LTI system is modelled as:
\begin{align} \label{eq_syst}
	\mb{x}(k+1) = A \mb{x}(k)+ \nu
\end{align}
with $k$ as the time index, $\mb{x} \in \mathbb{R}^n$ as the system state, $A$ as the $n$-by-$n$ system matrix, and $\nu$ as the system noise (assumed to be Gaussian). This discrete-time model might be obtained via the discretization of analog systems. Then, the digital system matrix $A$ has non-zero diagonals and is structurally full-rank. The measurements by the sensors are:
\begin{align} \label{eq_yi}
	\mb{y}_i(k) = C_i \mb{x}(k)+ \zeta_i
\end{align}
with $\mb{y}_i,C_i$ as the local output vector and output matrix at sensor $i$ and $\zeta_i$ as Gaussian measurement noise. The global measurement matrix of all sensors is defined by column concatenation of all measurements as, 
\begin{align}\label{eq_y}
	\mb{y}(k) = C \mb{x}(k)+ \zeta
\end{align}
with $\mb{y} \in \mathbb{R}^m$ as the global output vector and $m$-by-$n$ matrix $C$ as the global output matrix. As required for any estimation and filtering scenario the pair $(A,C)$ is observable. In the graph observability perspective, the LTI system is modelled as a system graph and the necessary and sufficient conditions for observability are defined, for example, in \cite{icassp2016,acc13_kar,liu-pnas}. In networked estimation perspective, given the sensor network $\mc{G}$, the LTI system~\eqref{eq_syst} needs to be observable to the group of sensors $\mc{V}$ so they can estimate the entire state $\mb{x}$. We make no local observability assumption at any sensor, but this \textit{distributed} observability is gained via the proper topological design of the sensor network and networked estimation protocol. The problem in this work is the resilient design such that after deleting $q$ (or less) number of nodes the distributed observability at remaining sensors is preserved. Similarly, after removing $q$ (or less) links the distributed observability is preserved. 

\subsection{Preliminaries on System Graph Theory}

\subsubsection{Graph Theory Notations and Concepts}
The sensor network is modelled as a graph $\mc{G}=\{\mc{V},\mc{E}\}$ with the set of nodes $\mc{V}=\{1,\dots,m\}$ representing the sensors and the set of (directed or undirected) links $\mc{E}$ denoting the communication links for data-sharing between sensor nodes. The network is said to be strongly connected (SC) if for every pair of nodes $i,j$ there is a directed path from $i$ to $j$ and vice versa. If the network is not SC, it can be decomposed into strongly-connected-components (SCCs). Define an SCC as a component (or subgraph) in which there exists a directed path between every pair of nodes.

\subsubsection{Relevant Notions on Link/Node-Connectivity}
Here, we provide some relevant notions on link/node connectivity. First, define a $(i,j)$-cut in a graph as a
subset $S_{ij} \subset \mc{V} $ such that if the nodes $S_{ij}$ are removed, the resulting graph contains no path from node $i$ to
node $j$. Let $\kappa_{ij}$ denote the size of the smallest $(i,j)$-cut between any two nodes $i$ and $j$. The graph $\mc{G}$ is
said to have node-connectivity $\kappa(\mc{G})$ if $\kappa_{ij} \geq \kappa$ for all nodes $i, j$. Similarly, the link-connectivity $e(\mc{G})$ is the minimum number
of links whose removal makes the graph disconnected.

Another relevant notion is algebraic connectivity which gives an estimate of the node/link-connectivity of the network. Given the 0-1 adjacency matrix $\mc{W}$ (as the structure/pattern of the \textit{weighted} adjacency matrix $W$) associated with the graph $\mc{G}$ and its diagonal degree matrix $D$ (with diagonals $D_{ii}$ as the degree of node $i$) define its Laplacian as $L=D-W$. The algebraic connectivity $\lambda_2(L)$ is defined as the second smallest eigenvalue of $L$ which is greater than zero for connected graphs. Then, the range of the node connectivity $\kappa(\mc{G})$ and link connectivity $e(\mc{G})$ satisfy \cite{godsil}
\begin{align} \label{eq_kappa}
	\lambda_2(L) \leq \kappa(\mc{G}) \leq e(\mc{G}) \leq d_{min} 
\end{align}
with $d_{min}$ as the minimum node degree of $\mc{G}$. The above relation implies that, in general, a given network is more resilient to link removal than node removal.   
An equivalent definition of node/link-connectivity is given below. 
\begin{defn}
	\textit{$q$-connectivity:} A network/graph is $q$-node-connected (resp. $q$-link-connected) if it remains SC after removal of up to $q$ nodes (resp. $q$ links).
\end{defn}
See Fig.~\ref{fig_Qgraph} to better illustrate the above definition. 
\begin{figure}[]
	\centering
	\includegraphics[width=3.5in]{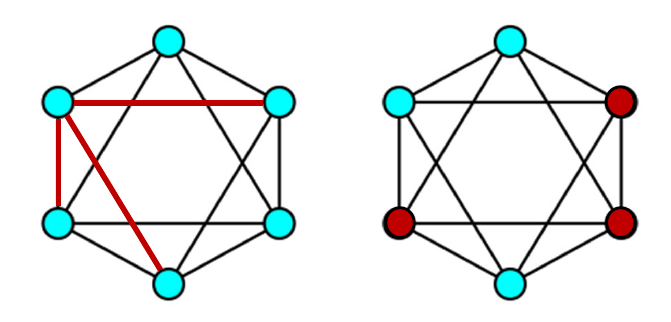}
	\caption{This figure shows an example $3$-node-connected and $3$-link-connected graph. in this example, by removing any set of (up to) $3$ links, the remaining graph is still strongly-connected; for example, by removing the red links the left graph holds its strong-connectivity. Similarly, by removing any set of (up to) $3$ nodes, the remaining reduced-order graph is still strongly-connected; for example, by removing  the red nodes (and their links) the right graph holds its strong-connectivity. Although this example shows an undirected network, the connectivity notions also hold for directed networks.}
	\label{fig_Qgraph}
\end{figure}
The algorithms to design $q$-node connected or $q$-link-connected networks can be found in \cite{galvez2021cycle,byrka2020breaching,cecchetto2021bridging,galvez2021approximation,gupta2011approximation}. This is also referred to as \textit{survivable network design} \cite{chen2022survivable,lau2007survivable}.

\subsubsection{Graph Observability via Structured Systems Theory}
Other relevant concepts in this work are structural observability of networks and observational equivalence. These concepts have been approached via a graph-theoretic (or structural) perspective as in \cite{liu-pnas,tnse18}.
\begin{defn} \label{defn_equiv}
	\textit{Observational Equivalence \cite{tnse18}:} Two sensor nodes $i,j$ are observationally equivalent if removing either
	of the two measurements $\mb{y}_i$ or $\mb{y}_j$ does not affect the $(A,C)$-observability, while removing
	both makes the pair $(A,C)$ unobservable.
\end{defn}

\begin{figure}[]
	\centering
	\includegraphics[width=3.5in]{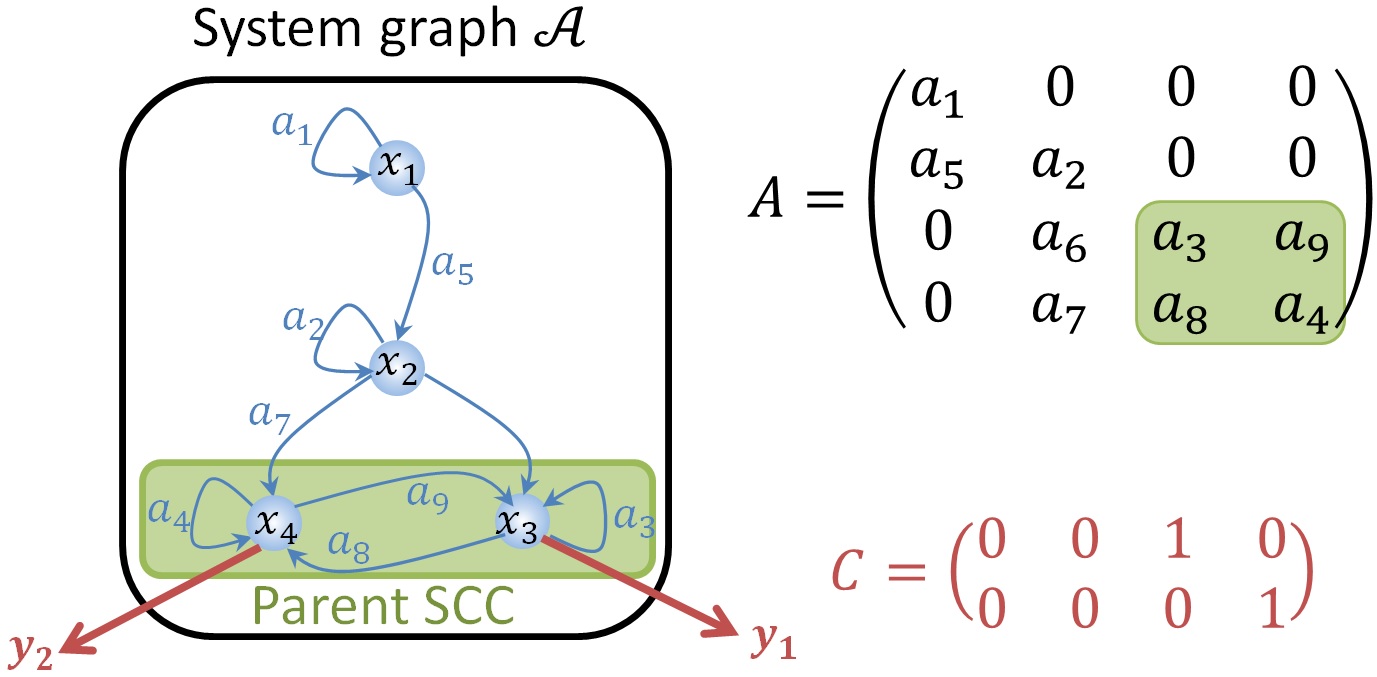}
	\caption{(Left) Example system graph $\mc{A}$ with one set of parent SCC and observationally equivalent outputs $y_1$ and $y_2$. (Right) The associated system matrix $A$ and output matrix $C$.}
	\label{fig_graph}
\end{figure}

Given a full-rank system matrix $A$ (e.g., via discretization of continuous-time systems), the sufficient condition for $(A,C)$-observability can be defined on the associated system digraph $\mc{A}$ and its SCCs. Given a system digraph $\mc{A}$, define its parent SCCs (or root SCC as referred to in \cite{liu-pnas}) as the SCCs with no outgoing links to other SCCs, see Fig.~\ref{fig_graph} for more illustrations. The SCC classification for a given system graph is determined by the depth-first search (DFS) algorithm\footnote{Note that DFS, particularly in its recursive form, might be susceptible to stack overflow issues and performance limitations due to the sheer size of the large-scale graphs.  An iterative implementation of DFS can be an alternative to mitigate stack overflow risks associated with recursion. Moreover, alternative techniques such as the improved version of \textit{Tarjan's} and \textit{Kosaraju's algorithm} \cite{pearce2016space} or other space-efficient algorithms \cite{hagerup2020space}  might be options to perform SCC classification.} \cite{algorithm}.       
\begin{lem} \label{lem_parent}
	Given a full-rank system matrix $A$, 
	\begin{enumerate}[label=(\roman*)]
		\item (at least) one output from every parent SCC in system digraph $\mc{A}$ is sufficient for $(A,C)$-observability,
		\item the state nodes in the same parent SCC are observationally equivalent.
	\end{enumerate}
\end{lem}
\begin{proof}
	The proof of (i) is given in \cite{acc13_kar,liu-pnas,icassp2016}. For the proof of (ii), recall from \cite[Theorem~1]{icassp2016} that the sufficient condition for observability of a full-rank system is to have a directed path from every node to (at least) one output-connected node. This follows from the output connectivity condition for structural observability and implies how observability is related to the existence of paths in the graph structure. From the definition, all nodes in the same parent SCC are mutually connected via a path. In other words, for any two state nodes $i,j$ in a parent SCC there exists a directed path from $i$ to $j$ and vice versa. This implies that the output from \textit{any} state node in the parent SCC results in output-connectivity of all other nodes in that SCC. This proves the observational-equivalence of all nodes in the same parent SCC.  
\end{proof}

\section{Resilient Estimator Design} \label{sec_res}
\subsection{The Proposed Estimation Technique}
The suggested distributed estimation scheme in this work is based on consensus/averaging on the information received from the neighbouring sensors. However, in contrast to \cite{acc13}, there is no observation-fusion step in our proposed scheme; thus, more relaxed network connectivity is needed in our new framework. This paves the way for the particular resilient design in terms of network topology.   
We propose the following distributed (or networked) estimator at every sensor $i$:

\small \begin{align}\label{eq_p} 
	\widehat{\mb{x}}_i(k|k-1) =& \sum_{j\in\mathcal{N}(i)} W_{ij}A\widehat{\mb{x}}_j(k-1|k-1),
	\\ 
	\widehat{\mb{x}}_i(k|k) =& \widehat{\mb{x}}_i(k|k-1) + K_i C_i^\top \left({y}_i(k)-C_i \widehat{\mb{x}}_i(k|k-1)\right), \label{eq_m}
\end{align} \normalsize
with $W_{ij}$ as the weight assigned to the estimation data received from sensor $j$. $W$ matrix is the consensus matrix and is row-stochastic\footnote{For a row-stochastic matrix $W$ we have $\sum_{j=1}^n W_{ij} = 1$. The simplest way to design this matrix is to set all the elements in each row of $W$ as $W_{ij} = \frac{1}{|\mathcal{N}(i)|}$ for $j\in\mathcal{N}(i)$. Another option is to design based on \textit{Metropolis-Hastings} method \cite{xiao2007distributed} as:
\begin{align} \nonumber
	W_{ij} = \left\{\begin{array}{ll}
	 \frac{1}{\max \{\mc{N}(i),\mc{N}(j)\}+1}, & (i,j) \in \mc{E}\\
		1-\sum_{j\in \mc{N}(i)} \frac{1}{\max \{\mc{N}(i),\mc{N}(j)\}+1}, & i=j\\
		0, & \mbox{otherwise}
	\end{array}\right.
\end{align}	
Other methods to design row-stochastic matrices are discussed in Section~\ref{sec_disc}.} with $\mc{W}$ as its structure. The a-priori and a-posteriori estimates at sensor $i$ are denoted by $\widehat{\mb{x}}_i(k|k-1)$ and $\widehat{\mb{x}}_i(k|k)$, $K_i$ is the local gain matrix, and $\mc{N}(i)$  denotes the set of neighbor sensors of $i$ over network $\mc{G}$. In the proposed distributed estimator, every sensor $i$ first performs information fusion on the \textit{estimates} received from its neighbouring sensors and, then, performs \textit{innovation-update} of this prediction by incorporating its own local measurement. This distributed setup mimics the structure of the Kalman filter; however, the main difference is the local consensus/averaging on the a-priori estimates in Eq.~\eqref{eq_p} which makes our estimator a distributed algorithm. Note that Kalman filter is a \textit{centralized} procedure, in which the central computing node receives all the information, and centrally estimates the system states. However, in our distributed setup, each sensor processes its local measurements and observation of the system and communicates with other sensors to refine its estimate. Thus, every sensor has its own estimate, which makes the distributed estimator more resilient to sensor failure. This is an advantage over the centralized Kalman filter which is prone to single-node-of-failuer. Furthermore, in the distributed setup, each sensor has a limited view of the overall state, i.e., the system is not fully observable to any sensor. The observability of the entire system relies on the collective observations of all sensors. In this direction, sensors share state estimates to enhance their system observability. This is in contrast to the Kalman filter, which needs global observability at the central computing node.
Moreover, as compared to the other distributed setups, it is clear that only one step of data-fusion is performed in Eq. \eqref{eq_p} between two consecutive time-steps $k$ and $k-1$. This clarifies that the estimator is in single time-scale, in contrast to many steps of data-fusion (referred to as the \textit{consensus loop}) in \cite{he2020secure,battilotti2021stability,olfati2007distributed}.  
Further, note that in the innovation-update step~\eqref{eq_m}, there is no information-fusion on measurement data and, thus, there is no need to share the measurements. As claimed in \cite{acc13}, a \textit{hub-based} network with more connectivity is needed for this step. Simplifying this step, thus, is an improvement over the proposed protocol in \cite{acc13} in terms of relaxing the required sensor-network connectivity. As we will see later in this paper, this is a key aspect towards survivable network design.

Next, we investigate the stability of the error dynamics for this estimator. Define the local error as $\mb{e}_i(k) := \mb{x}(k)-\widehat{\mb{x}}_i(k|k)$ and the global error as the column concatenation of all local errors as $\mb{e}(k) := (\mb{e}_1(k),\cdots,\mb{e}_m(k))^\top $. The error dynamics at every sensor $i$ of the proposed estimator takes the following form:
\begin{align}
	\nonumber
	\mb{e}_i(k) =&\mb{x}(k) - \widehat{\mb{x}}_i(k|k-1) -K_i C_i^\top \left({y}_i(k)-C_i \widehat{\mb{x}}_i(k|k-1)\right)
	\\ 
	=&\mb{x}(k) - \sum_{j\in\mathcal{N}(i)} W_{ij}A\widehat{\mb{x}}_i(k-1|k-1) \nonumber \\
	&-K_i C_i^\top \left({y}_i(k)-C_i \widehat{\mb{x}}_i(k|k-1)\right) \nonumber
\end{align}
Substituting Eq.~\eqref{eq_syst}-\eqref{eq_y} we have,

\small \begin{align}
	\mb{e}_i(k) &= A\mb{x}(k-1) + \nu - \sum_{j\in\mathcal{N}(i)} W_{ij}A\widehat{\mb{x}}_i(k-1|k-1) \nonumber \\
	&- K_i C_i^\top \left( C_i \mb{x}(k)+ \zeta_i-C_i \sum_{j\in\mathcal{N}(i)} W_{ij}A\widehat{\mb{x}}_i(k-1|k-1)\right) \nonumber 
	\\ \nonumber
	&= A\mb{x}(k-1) + \nu - \sum_{j\in\mathcal{N}(i)} W_{ij}A\widehat{\mb{x}}_i(k-1|k-1) \nonumber \\ 
	&-K_i \Bigl(C_i^\top C_i(A\mb{x}(k-1) + \nu) + \zeta_i \nonumber \\
	&-C_i^\top C_i\sum_{j\in\mathcal{N}(i)} W_{ij}A\widehat{\mb{x}}_i(k-1|k-1)\Bigr). \label{eq_err_semi}	 
\end{align}\normalsize 
Recalling that the $W$ matrix is stochastic, we have ${A\mb{x}(k-1)=\sum_{j\in \mathcal{N}(i)} W_{ij}A\mb{x}(k-1)}$. By substituting this in the above error equation \eqref{eq_err_semi} and substituting  $\mb{e}_j(k-1) := \mb{x}(k-1)-\widehat{\mb{x}}_j(k-1|k-1)$, we have
%\begin{align}
%	\mb{e}_{k}^i &=\sum_{j\in \mathcal{N}_\beta(i)} W_{ij}A\mb{x}_{k-1} - \sum_{j\in\mathcal{N}_\beta(i)} W_{ij}A\widehat{\mb{x}}^j_{k-1|k-1} \nonumber \\ &- K_i \sum_{j\in \mathcal{N}_\alpha(i)} \mb{c}_j \mb{c}^\top_j\Bigl(\sum_{j\in \mathcal{N}_\beta(i)}  W_{ij}A\mb{x}_{k-1} - \sum_{j\in \mathcal{N}_\beta(i)}   W_{ij}A\widehat{\mb{x}}^j_{k-1|k-1} \Bigr) + \pmb{\nu}_{k-1}-K_i\sum_{j\in \mathcal{N}_\alpha(i)} \mb{c}_j \zeta^j_k +
%	\mb{c}_j \tau^j_{k}+ \mb{c}_j \mb{c}^\top_j\pmb{\nu}_{k-1} \nonumber \\	
\begin{align} \nonumber
	\mb{e}_i(k) &=\sum_{j\in \mathcal{N}(i)} W_{ij}A\mb{e}_j(k-1) \\ &-K_i  C_i^\top C_i\sum_{j\in \mathcal{N}(i)}  W_{ij}A\mb{e}_j(k-1)+ \eta_i(k), \label{eq_err_i}
\end{align}
with ${\eta_i(k) := \nu(k-1)-K_i (C_i \zeta_i(k) +
	C_i^\top C_i \nu(k-1))}$ collecting the noise terms.
The global error dynamics is,
\begin{align} \nonumber
	\mb{e}(k) &= (W\otimes A - KD_C(W\otimes A))\mb{e}(k-1) +
	{\eta}(k)  \\
	&= \widehat{A}\mb{e}(k-1) +
	\eta(k), \label{eq_err1} 
\end{align}
with $\eta(k)$ as the column concatenation of all the noise terms $\eta_i(k)$, ``$\otimes$'' as the Kronecker product,
%\begin{align}
%\eta(k) &= \mathbf{1}_N \otimes \pmb{\nu}_{k-1} 
%- K D_C(\mathbf{1}_N \otimes \pmb{\nu}_{k-1}) - K\overline{D}_C\pmb{\zeta}_{k} -K\overline{D}_C\pmb{\tau}_{k},
%\label{eq_eta}
%\end{align}
${\widehat{A} := W\otimes A - KD_C(W\otimes A)}$,
${K := \mbox{blockdiag}[K_i]}$,   ${D_C := \mbox{blockdiag}[C_i^\top C_i]}$. 
%and  $\overline{D}_C \triangleq (U \otimes \mb{1}_n) \circ (\mb{1}_N \otimes C^\top  ) $ with ``$\circ$'' and ``$\otimes$'', respectively, as the entrywise (Hadamard) and Kronecker product.
%\begin{ass} \label{ass_sc}
%    The sensor network $\mc{G}$ is strongly-connected.
%\end{ass}

Given the strong-connectivity of the sensor network $\mc{G}$, we prove that the distributed estimator \eqref{eq_p}-\eqref{eq_m} is distributed observable and can be made Schur stable by proper design of the gain matrix $K$. 
\begin{lem} \label{lem_sc}
	The error dynamics \eqref{eq_err1} is Schur stabilizable if the matrix $W$ is irreducible, i.e., the sensor network $\mc{G}$ is SC.
\end{lem}
\begin{proof}
	From the Kalman theorem \cite{bay} the linear system \eqref{eq_err1} is Schur stabilizable if the pair $(W\otimes A, D_C)$ is observable. This relates the stability of error dynamics \eqref{eq_err1} to the observability of distributed system matrix $W\otimes A$ and matrix of shared measurements $D_C$. Recall that, using graph theory, the (structural) observability of the pair $(W\otimes A, D_C)$ can be discussed as the \textit{Kronecker network product} of the sensor network $\mc{G}$ and the system graph $\mc{A}$. Given that the system $A$ is full-rank, its system digraph $\mc{A}$ is cyclic. Then, using \cite[Theorem~4]{tsipn}, the minimal sufficient condition for observability of this graph product is that the network $\mc{G}$ be SC. In other words, the sensor network $\mc{G}$ being SC guarantees the observability of the Kronecker product network and its associated system-measurement pair, i.e., strong connectivity of $\mc{G}$ (i.e., irreducibility of $W$) implies $(W\otimes A, D_C)$-observability. This completes the proof. 
\end{proof}

Given the necessary condition for distributed observability, the next step is to design the feedback gain matrix $K$. Recall that for distributed and localized filtering, the feedback gain $K$ needs to be \textit{block-diagonal}. This ensures that every sensor node uses its own information feedback for estimation. Such a block-diagonal gain matrix can be designed by using linear-matrix-inequality (LMI) based on cone-complementary algorithms proposed in \cite{el1997cone,usman_cdc:11,mangasarian1995extended}. Given that the distributed estimator \eqref{eq_p}-\eqref{eq_m} is generically observable (as proved in Lemma~\ref{lem_sc}), a structured block-diagonal feedback gain matrix, $K$, needs to be designed for the distributed observer. We need to find $Q \succ 0$ (`$\succ$' denotes positive-definiteness) for $\widehat{A} := W \otimes A -K D_C (W\otimes A)$ such that,
\begin{align}
	Q-\widehat{A}Q\widehat{A} \succ 0
\end{align}	
which is equivalent to the following,
\begin{equation} \label{eq_lmi0}
	\begin{aligned}
		& ~~ \left( \begin{array}{cc} Q&\widehat{A}^\top Q\\ Q\widehat{A}&Q\\ \end{array} \right) \succ 0
	\end{aligned}
\end{equation}
The above LMI cannot be directly solved for the block-diagonal gain matrix. An equivalent solution is given in the literature to solve this which we recall here. As stated in \cite{pajic2010wireless}, to have $\rho(\widehat{A})<1$, there must exist $Q,R \succ 0$ such that
\begin{equation} \label{eq_lmi1}
\left( \begin{array}{cc} Q&\widehat{A}^\top\\ \widehat{A}&R\\ \end{array} \right) \succ 0
\end{equation}
with~$Q=R^{-1}$. This constraint is non-convex. Therefore, using the relaxation strategy in \cite{el1997cone}, we replace $Q=R^{-1}$ with its linear equivalent as follows. The matrices,~$Q,R \succ 0$, satisfy~$Q=R^{-1}$ as the optimizer of the following LMI:
\begin{equation}
	\begin{aligned}
		& ~\displaystyle\min ~ \mathbf{trace}(QR)~\mbox{subject~to} ~\left( \begin{array}{cc}  Q&I\\ I&R\\ \end{array} \right) \succeq 0,
	\end{aligned}
\end{equation} 
The above can be summarized as the following LMI:
\begin{equation} \label{eq_LMI2}
	\begin{aligned}
		\displaystyle
		\min
		~~ &  \mathbf{trace}(QR) \\
		\text{s.t.}  ~~& Q,R\succ 0,\\ ~~ & \left( \begin{array}{cc} Q&\widehat{A}^\top\\ \widehat{A}&R\\ \end{array} \right) \succ 0,~ \left( \begin{array}{cc} Q&I\\ I&R\\ \end{array} \right) \succ 0,\\
		~~ & K\mbox{~is~block-diagonal}.\\
	\end{aligned}
\end{equation} 
where $\mathbf{trace}(QR)$ can be replaced with the linear approximation~$\Phi _{lin}(Q,R)=\mathbf{trace}(Q_0R+R_0Q)/2$ \cite{el1997cone}. The following iterative algorithm is used to minimize this trace under the constraints.
\begin{algorithm} \label{alg_1}
		\textbf{Initialization:} feasible points~$Q_0,R_0,K$\; 
		\While{termination criteria NOT hold}{
			Minimize~$\mathbf{trace}(Q_tR + R_tQ)$ under the constraints in Eq.~\eqref{eq_LMI2} and find new~$Q,R,K$\;
			If~$\rho (\widehat{A}) <1$ terminate, otherwise set~$R_{t+1}=R,~Q_{t+1}=Q$ and run the above step for next iteration~$t=t+1$\;
			Set $t \leftarrow t+1$ \;
		}
		\textbf{Return: $Q,R,K$}  \;	
		\caption{Iterative calculation of block-diagonal gain $K$}
\end{algorithm}
\\ 
As shown in \cite{el1997cone}, $\mathbf{trace}(Q_tR + R_tQ)$ in Algorithm~\ref{alg_1} is a non-increasing sequence that converges to~$2mn$. In this regard, a stopping criteria in Algorithm~\ref{alg_1} is established as reaching within~$2mn + \epsilon$ of the
trace objective. This is set as the termination criteria of the loop. We use the CVX toolbox in MATLAB to solve this iterative LMI algorithm. 
%We refer interested readers to~\cite{usman_cdc:11,5717159,rami:97, siljak08} for more details.
It is known that such cone-complementary algorithms are of polynomial-order complexity and, thus, are adaptable for large-scale applications.

Note that Lemma~\ref{lem_sc} relaxes the connectivity requirement 
in previous works \cite{acc13,kar2013consensus,das2015distributed} to strong-connectivity with no assumption on the local observability of any sensor. Recall that \cite{acc13,kar2013consensus,das2015distributed} mandate more than strong-connectivity for convergence and stability. The work \cite{acc13} requires a hub-based network and measurement sharing in \eqref{eq_m} for stability. The works \cite{kar2013consensus,das2015distributed} require observability in the neighbourhood of every sensor over the sensor network $\mc{G}$. On the other hand, the works \cite{he2020secure,battilotti2021stability,olfati2007distributed} perform many consensus steps (more than the diameter of $\mc{G}$) as the inner loop to relax the observability condition. These works mandate very fast communication devices to gain observability (via excessive data-sharing) between every two steps of system dynamics. This is in contrast to the one-step consensus in estimator \eqref{eq_p}-\eqref{eq_m}. In this sense, our work relaxes the network connectivity and communication requirement in many existing distributed estimation literature. 
Next, for \textit{resilient} (or\textit{ $q$-redundant}) estimator design, the following assumption is made. 

\begin{ass} \label{ass_ql}
	The sensor network $\mc{G}$ is $q$-link-connected.
\end{ass}
To satisfy assumption~\ref{ass_ql} algorithms in \cite{galvez2021cycle,byrka2020breaching,cecchetto2021bridging,galvez2021approximation,gupta2011approximation} can be used. Given any network, these algorithms recover its $q$-connectivity and provide a method to improve the connectivity of the resulting network and make it $q$-link-connected for any desired $q$.    

\begin{theorem}
	The distributed estimator \eqref{eq_p}-\eqref{eq_m} under Assumption~\ref{ass_ql} is resilient to the removal of up to $q$ communication links.
\end{theorem}

\begin{proof}
	First, recall from Lemma~\ref{lem_sc} that the error dynamics of Eq. \eqref{eq_p}-\eqref{eq_m} is Schur stabilizable if the sensor network $\mc{G}$ is SC. From Assumption~\ref{ass_ql} and the definition of $q$-link connectivity, the network $\mc{G}$ remains SC after deletion of up to $q$ number of links. Therefore, $(W\otimes A, D_C)$-observability is preserved after link removal and the design is resilient. This completes the proof.  
\end{proof}	
Recall that for an observable pair $(W\otimes A, D_C)$, using the linear-matrix-inequality (LMI) design in \cite{usman_cdc:11}, one can find block-diagonal gain matrix $K$ to stabilize \eqref{eq_err1}.

Next, we discuss $q$-node connectivity and node-removal resiliency. Note that, in case of sensor node removal, even the centralized $(A,C)$-observability might be lost after node deletion. This is where the concept of observational equivalence plays a key role. To preserve observability up to $q+1$ observationally equivalent sensor nodes are needed from every set of parent SCC in $\mc{A}$. Therefore, for the estimator design the underlying sensor network includes (at least) $q+1$ number of sensors with measurements from every parent SCC in $\mc{A}$. Then, one can design a $q$-node-connected network to gain network resiliency and preserve strong connectivity against node deletion. These are summarized in the next assumption. 

\begin{ass} \label{ass_qn}
	The following assumptions on the sensor network $\mc{G}$ holds:
	\begin{enumerate}[label=(\roman*)]
		\item It includes $q+1$ observationally equivalent sensor measurements from every parent SCC in $\mc{A}$.
		\item It is $q$-node-connected.
	\end{enumerate}
\end{ass}
To satisfy Assumption~\ref{ass_qn}-(i) first using the DFS algorithm \cite{algorithm} the parent SCCs are defined, and then  $q+1$ sensor measurements from every parent SCC are set. Then, to satisfy Assumption~\ref{ass_qn}-(ii), algorithms in \cite{galvez2021cycle,byrka2020breaching,cecchetto2021bridging,galvez2021approximation,gupta2011approximation} are used to design and recover $q$-node-connectivity of the sensor network for any desired $q$. This is also discussed in Algorithm~\ref{alg_ac}. Note that this algorithm is of polynomial-order complexity, which makes it applicable in large-scale applications.

\begin{algorithm} 
	\KwData{ system matrix $A$, system digraph $\mc{A}$}
	\KwResult{ sensor network $\mc{G}$, estimate $\widehat{\mb{x}}_i(k|k)$ }
	{\textbf{Initialization:} $k=1$ and random state initialization\;
		Define output matrix $C$ with $q+1$ outputs from every parent SCC in $\mc{A}$\;
		Assign every output to one sensor $i$\;
		Design a $q$-node-connected sensor network $\mc{G}$ via techniques in \cite{galvez2021cycle,byrka2020breaching,cecchetto2021bridging,galvez2021approximation,gupta2011approximation}\;
		Define row-stochastic consensus matrix $W$\;
		Find $K$ via the iterative LMI design in Algorithm~\ref{alg_1}\;
	}
	\While{termination criteria NOT hold\;
	}{Sensor $i$ receives $\widehat{\mb{x}}_j(k-1|k-1)$ from neighbor sensors $j\in \mc{N}(i)$\;
		Sensor $i$ updates $\widehat{\mb{x}}_i(k|k)$ via dynamics~\eqref{eq_p}-\eqref{eq_m}\;
		$k\leftarrow k+1$\;
	}
	\caption{\textsf{The Resilient Estimator Design}} \label{alg_ac} 
\end{algorithm}

It should be emphasized that the distributed estimation steps in Algorithm~\ref{alg_ac} are innovative, and only the $q$-node-connected redundant design is via the results in \cite{galvez2021cycle,byrka2020breaching,cecchetto2021bridging,galvez2021approximation,gupta2011approximation}. In particular, the iterative distributed estimation updates via dynamics \eqref{eq_p}-\eqref{eq_m} advance the algorithm beyond existing methods. No existing work in the literature considers our proposed single time-scale estimation which relaxes the local observability assumption while removing the consensus loop on sensors. This significantly reduces the communication and computation load on sensors and makes the algorithm more practical in large-scale real-world setups. In particular,  dynamics \eqref{eq_p}-\eqref{eq_m} only requires strong-connectivity of the network, which reduces information-sharing in the neighborhood of sensors and further allows for $q$-redundant resilient design via graph-theoretic approaches. Many existing works require more than strong-connectivity and do not allow for redundant design via the techniques in \cite{galvez2021cycle,byrka2020breaching,cecchetto2021bridging,galvez2021approximation,gupta2011approximation}.

\begin{theorem}
	Under Assumption~\ref{ass_qn} and Algorithm~\ref{alg_ac} the distributed estimation \eqref{eq_p}-\eqref{eq_m} is resilient to removal of up to $q$ sensor nodes.
\end{theorem}
\begin{proof}
	The proof is twofold. First, we show that, under removal of $q$ sensor nodes, centralized $(A,C)$-observability holds and then we prove that distributed $(W \otimes A, D_C)$-observability holds. Following from Assumption~\ref{ass_qn}-(i) and Algorithm~\ref{alg_ac} the sensor network includes $q+1$ outputs from every parent SCC. Therefore, after the removal of any set of $q$ sensors, (at least) one output from every parent SCC remains and guarantees $(A,C)$-observability. This follows from Lemma~\ref{lem_parent}. Next, following from Assumption~\ref{ass_qn}-(ii) and Algorithm~\ref{alg_ac}, since the network $\mc{G}$ is $q$-node-connected it holds strong connectivity after removal of any set of $q$ sensor nodes. The strong-connectivity of $\mc{G}$ is sufficient for distributed $(W \otimes A, D_C)$-observability\footnote{Note that global $(A,C)$-observability is a necessary condition for any distributed or centralized estimation setup. On the other hand, distributed $(W \otimes A, D_C)$-observability defines the sufficient condition on the sensor network $\mc{G}$, i.e., for observability of system matrix $A$ to any sensor over the network $\mc{G}$. }. Then, the estimator  \eqref{eq_p}-\eqref{eq_m} is Schur stable. This completes the proof.    
\end{proof}	

\subsection{Discussions} \label{sec_disc}
In this subsection, some noteworthy comments are provided to highlight the advantages and limitations of the proposed distributed resilient setup.
\begin{table*} [bpt!] 
	\centering
	\caption{Comparing estimation methods in terms of assumption, design, and number of communication and consensus (computation) iterations for a sensor network of size $m$. }
		\label{tab_compare}
		\begin{tabular}{|c|c|c|c|c|c|} 
			\hline
			Literature &	time-scale & observability & communication & consensus & resilient design  \\
			\hline
			this work &	single & global-$(A,C)$ &  $m \times 1$ & $1$ & $\checkmark$
			\\			
			\hline
			\cite{kar2013consensus,das2015distributed,Silm2020A,khan2010connectivity,chen2018internet} &	single &  local-$(A,C_i)$ &  $3m \times 1$   & $1$ & -
			\\			
			\hline 
			\cite{he2020secure} &	double & global-$(A,C)$ &  $m \times L$    & $L$ & $\checkmark$
			\\
			\hline 
			\cite{battilotti2021stability,olfati2007distributed} &	double & global-$(A,C)$ &  $m \times L$    & $L$ & -
			\\
			\hline
	\end{tabular}
\end{table*}
\begin{itemize}
	\item There exist distributed observer-based techniques to detect and isolate faulty (or attacked) sensors in distributed estimation setup. These algorithms allow to \textit{locally} find the biased sensory information over the network and prevent the cascade of faulty data to the rest of the sensor network. See \cite{teixeira2014distributed,tcns20,hajshirmohamadi2019event,davoodi2013distributed,ferrari2011distributed} for some example localized FDI techniques.
	\item The proposed distributed estimator is in single time-scale with no inner consensus loop (as illustrated in Fig.~\ref{fig_scale}). In other words, based on the estimator \eqref{eq_p}-\eqref{eq_m}, every sensor node only performs one \textit{single} step of consensus/averaging and information-exchange between every two consecutive samples of system dynamics. This considerably reduces the computation and communication load on sensors. Therefore, the proposed resilient estimator is more efficient in terms of computation and communication as compared to the existing double time-scale estimators \cite{he2020secure,battilotti2021stability,olfati2007distributed}. 
	Moreover, The proposed estimator \eqref{eq_p}-\eqref{eq_m} only requires distributed $(W\otimes A, D_C)$-observability with no need for local observability in the neighbourhood of any sensor. This relaxes the observability and subsequently connectivity requirement as compared to many existing single time-scale estimators \cite{kar2013consensus,das2015distributed,Silm2020A,khan2010connectivity,chen2018internet}. This sets the foundation for our resilient $q$-redundant observer design. Detailed comparison with these existing works is presented in Table~\ref{tab_compare}.
   
	\item Cost-optimization techniques can be used to optimally design the $q$-redundant sensor network. In this approach, every link is associated with a cost and certain assignment techniques can be used to optimally assign the linking over the network. 
	For example, wireless sensor networks often operate in resource-constrained environments, where factors such as limited battery life, processing power, and communication bandwidth are critical \cite{ferentinos2007adaptive,xu2005optimal}. Cost-optimal design of such networks aims to maximize the utility of the limited communication resources while maintaining sufficient connectivity for resilient estimation. 
	\item The convergence of the proposed estimator does not depend on the exact consensus weights in $W$ matrix as far as its row-stochasticity is satisfied. Therefore, different approaches for the stochastic weight design can be adopted; namely, Metropolis-Hastings \cite{xiao2007distributed}, quantized averaging \cite{aysal2008distributed}, or ADMM-based \cite{erseghe2011fast}.    
	\item The proposed resilient estimator design, in contrast to the centralized estimation techniques \cite{Ding2022Resilient}, is \textit{scalable}. Therefore, one may accommodate varying deployment sizes and spatial coverage. This allows application in large-scale systems as in smart cities, IoT, or industrial automation \cite{chen2018internet}, where scalable sensor networks can accommodate the diverse needs of different applications.  
	\item Although efficient algorithm exists to design $q$-node-connected, determining whether a graph is $q$-node-connected is an NP-complete problem \cite{Qnode}. Several \textit{heuristic} algorithms and approximation techniques exist to find approximate solutions or to verify $q$-node-connectivity efficiently, see \cite{cheriyan2001rooted} and the related work on Hamiltonian cycles \cite{Kang1991Hamiltonian} for example. Algorithms for finding $q$-link-connected graphs typically involve exploring all possible node pairs and finding $q$ disjoint paths between them. This can be achieved using graph-theoretic techniques such as DFS or breadth-first search (BFS) combined with backtracking or dynamic programming \cite{algorithm,Boyd2022Approximation}.    
\end{itemize}

\section{Simulations} \label{sec_sim}
\begin{figure}[]
	\centering
	\includegraphics[width=3.5in]{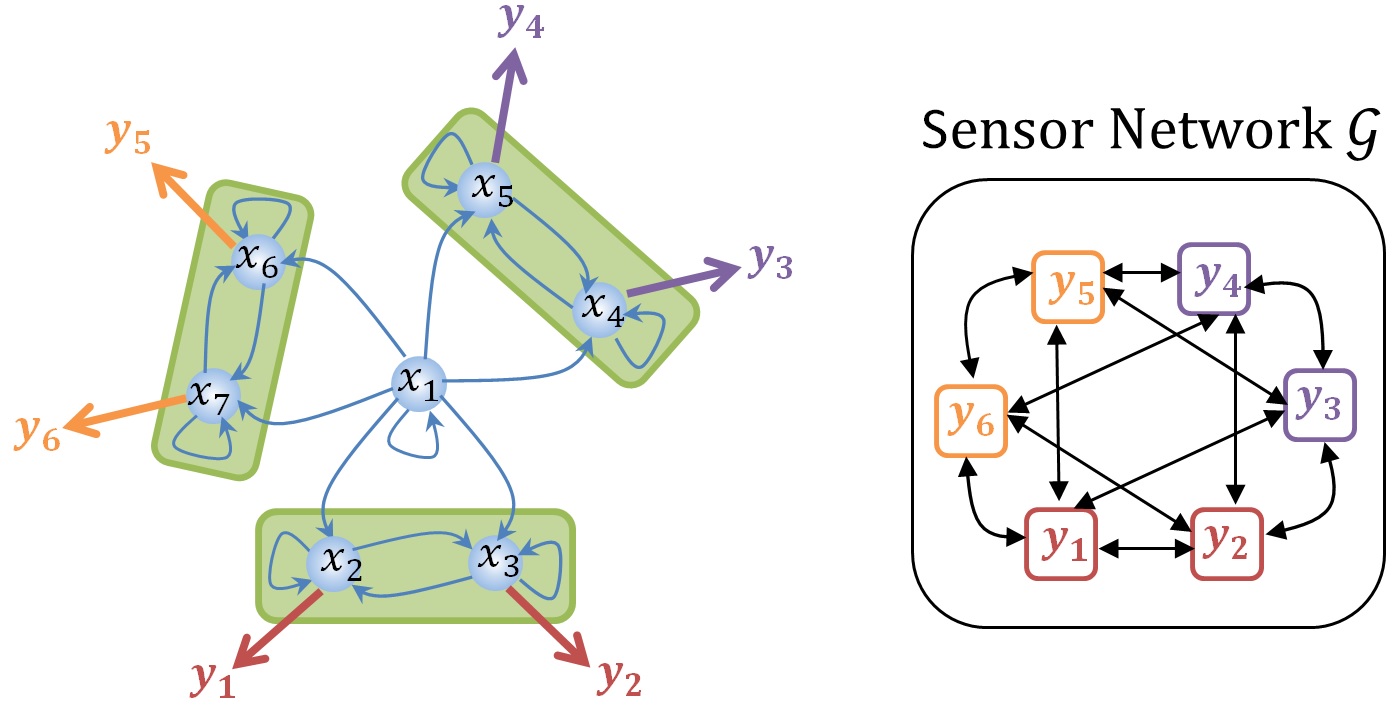}
	\caption{(Left) This figure shows an example system graph $\mc{A}$ with three sets of parent SCCs of size two and a pair of sensor measurements from each parent SCC; recall from Definition~\ref{defn_equiv} that outputs of the same colour are observationally equivalent. (Right) The example sensor network to monitor this system graph is $2$-node-connected and $2$-link-connected. Thus, the distributed estimation network is resilient to the failure of $1$ sensor node (due to the size of parent SCCs and equivalent set of sensor outputs) or the failure of $2$ communication links. }
	\label{fig_graph_net}
\end{figure}
%\subsection{Simulation in the Absence of Failure}
We consider the system graph $\mc{A}$ in Fig.~\ref{fig_graph_net}-(Left) for simulation and set its system parameters as random values. The system and measurement noise are set as Gaussian noise $\mc{N}(0,0.1)$. The initial state values for the estimator are set randomly. As shown in the figure, this system graph includes three sets of parent SCCs each with one pair of observationally equivalent measurements. Each measurement is assigned to a sensor. The sensor network $\mc{G}$ to track the states of the system graph is shown in Fig.~\ref{fig_graph_net}-(Right). The consensus weights in $W$ are chosen randomly while satisfying row-stochasticity. The local (block-diagonal) estimator gain is designed based on iterative LMI in Algorithm~\ref{alg_1}. To solve this optimization problem we use the CVX toolbox in MATLAB. The proposed distributed estimation setup in Algorithm~\ref{alg_ac} is used over the weight-stochastic network $\mc{G}$ to estimate all the states of the underlying system $A$. By sharing estimation data over  $\mc{G}$,  mean-square-errors (MSE) at all sensor nodes under the proposed estimator \eqref{eq_p}-\eqref{eq_m} are shown in Fig. \ref{fig_sim_no}. 

\begin{figure}[]
	\centering
	\includegraphics[width=1.75in]{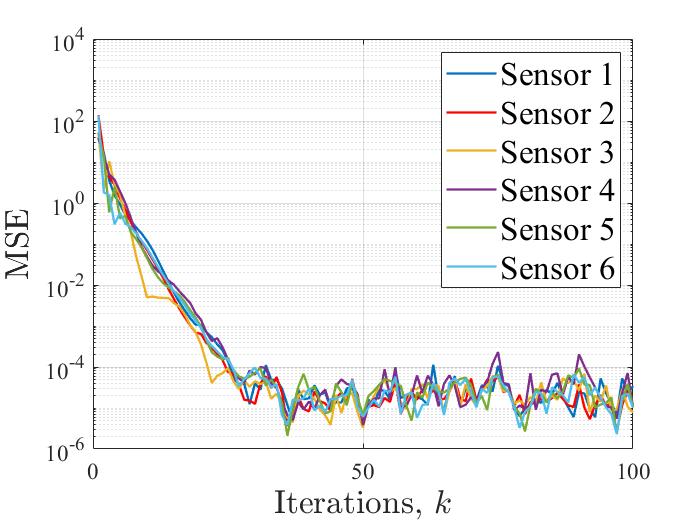}
	\caption{The evolution of MSE at six sensors tracking the state of system graph in Fig.~\ref{fig_graph_net}. }
	\label{fig_sim_no}
\end{figure}

\textbf{Comparison with Literature:} Next, we simulate the double time-scale distributed estimator in \cite{he2020secure} for comparison over the same setup. For this estimator, following Fig.~\ref{fig_scale}, the sensors perform $L$ extra steps of communication, averaging, and consensus-fusion between every two consecutive iterations $k$ and $k+1$ of the system dynamics. The parameters of \cite{he2020secure} are set as $\alpha = 0.5$, $\beta = 0.5$, and $k_{sat}=1$. We considered two values of $L=4$ and $L=20$ steps and presented the simulations in Fig~\ref{fig_mse_comp}. As claimed in \cite{he2020secure} and from the simulations, the error performance closely depends on the number of averaging steps $L$, which needs to be sufficiently large. Recall that, for this scenario, sensors need to share data and perform consensus $L$ \textit{times faster} than the rate of the system dynamics. This clearly requires more costly communication/computation accessories at the sensor network. Therefore, there is a trade-off on the parameter $L$ in terms of error performance and communication/computation load on sensors. As compared to our setup with only $1$ step of data sharing and consensus, the $L$ times more number of averaging and consensus in \cite{he2020secure} considerably increases the computational complexity of the algorithm and information exchange among the sensors.   
\begin{figure}[]
	\centering
	\includegraphics[width=1.73in]{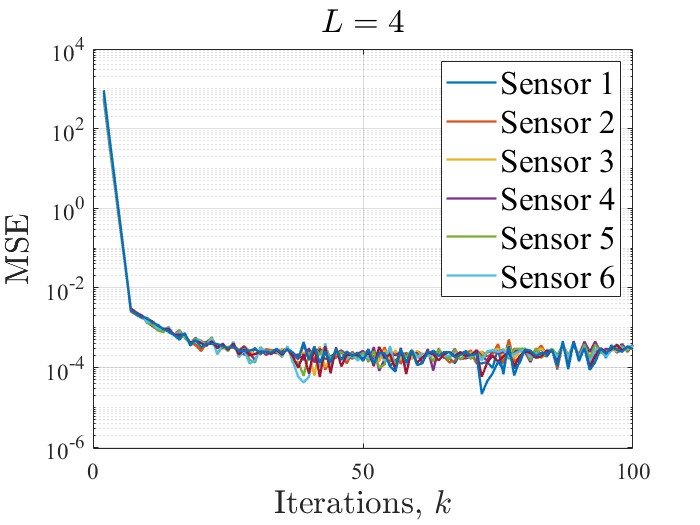}
	\includegraphics[width=1.73in]{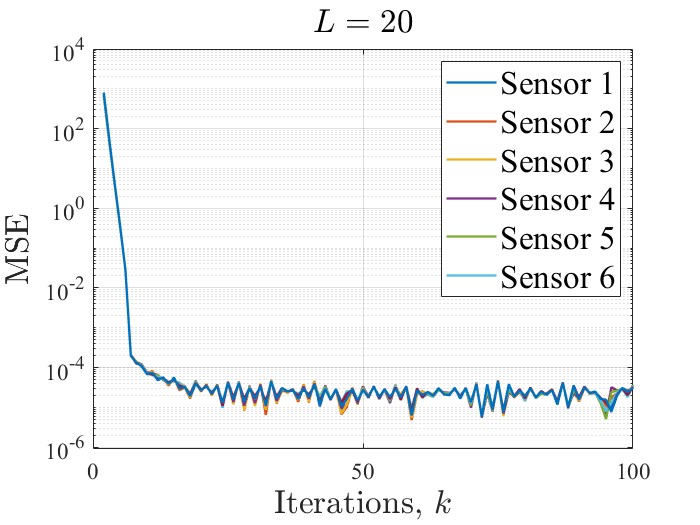}
	\caption{The time-evolution of MSE at sensors under double time-scale estimator in \cite{he2020secure}: (Right) for $L=4$ steps of consensus/averaging and (Left) for $L=20$ steps of consensus/averaging.}  
	\label{fig_mse_comp}
\end{figure}

\textbf{Simulation under node/link Failure:} Note that, from the definition, the sensor network $\mc{G}$ is $3$-node-connected and $3$-link-connected. Therefore, it preserves its connectivity after the removal of up to $3$ links or $3$ sensor nodes. The following remark is noteworthy.
\begin{rem}
	Although the sensor network $\mc{G}$ is $3$-node-connected, it is resilient to failure/removal of only one sensor node. This is because there are only two observationally equivalent sensors with outputs from parent SCCs. In other words, for resilience to node removal, both parts (i) and (ii) of Assumption~\ref{ass_qn} must be satisfied. 
\end{rem}
We redo the simulation for two cases: (i) after deleting $3$ links $\{(1,2),(2,4),(3,4)\}$ (randomly chosen) and (ii) after randomly removing one sensor (and its communication links) from the network $\mc{G}$. The simulations for these are shown in Fig.~\ref{fig_sim_del}. As it can be seen, the MSE evolution is stable for both cases which verifies our resilient design. This implies that, although some data-sharing channels and sensor observations are lost, the remaining sensor network and measurements are sufficient for estimating the underlying system. Note that the convergence rate depends on $\rho(\widehat{A})$ which in turn depends on both $K$ matrix and $W$ matrix. Therefore, the rate of error decay at sensors may change for different choices of these two matrices.   
\begin{figure}[]
	\centering
	\includegraphics[width=1.73in]{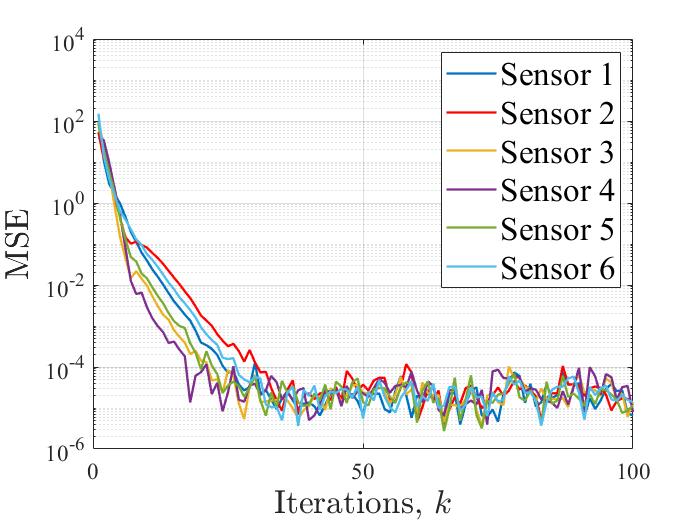}
	\includegraphics[width=1.73in]{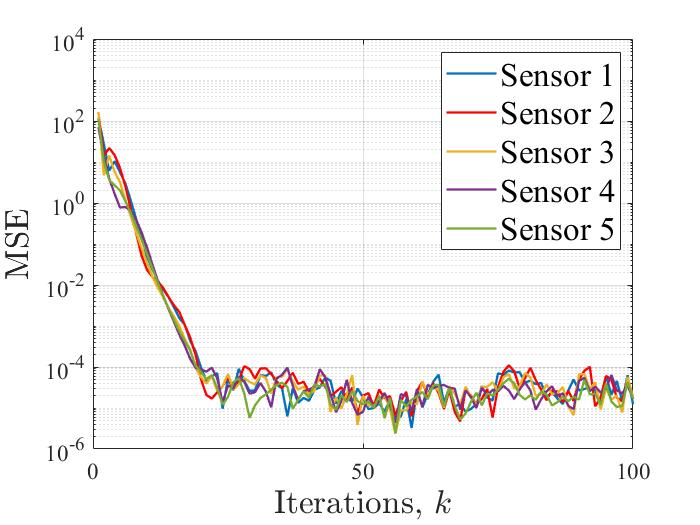}
	\caption{The evolution of MSE at sensors tracking the same system graph in Fig.~\ref{fig_graph_net}: (Left) after three link removal and (Right) after one sensor node removal. The stable MSE approves our resilient design.  }
	\label{fig_sim_del}
\end{figure}

\textbf{Large-scale Example:} 
To substantiate the scalability of the proposed distributed estimator, we consider a large-scale example system composed of the $10$ replications of the system graph example $\mc{A}$ in Fig.~\ref{fig_graph_net}-(Left). Therefore, this large-scale system graph includes $30$ parent SCCs. We considered $m=60$ sensors over a resilient $3$-node/link-connected sensor network to estimate the $n=70$ states of the underlying dynamical system with $A \in \mathbb{R}^{70 \times 70}$. The time evolution of the MSE at all sensors is shown in Fig.~\ref{fig_sim_large}. Recall that both cone-complementary-based Algorithm~\ref{alg_1} for designing $K$ matrix and the distributed estimation algorithm~\ref{alg_ac} are of polynomial order complexity. This makes it applicable in large-scale setups like the one provided here.
\begin{figure}[]
	\centering
	\includegraphics[width=1.75in]{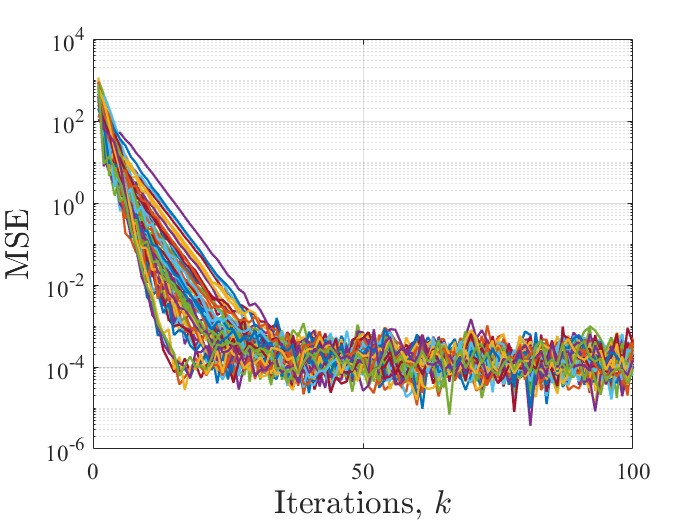}
	\caption{The evolution of MSE at $m=60$ sensors tracking the state of a large-scale system graph example of $n=70$ states.} 
	\label{fig_sim_large}
\end{figure}

\section{Conclusions} \label{sec_con}
\subsection{Concluding Remarks}
This work presents a distributed resilient estimator based on the notions of $q$-connectivity and observational equivalence. We first derive the distributed observability condition of the proposed algorithm and relate it to the strong-connectivity of the sensor network. This considerably relaxes the connectivity requirement in many existing literature. Then, $q$-redundant-connectivity is considered for the resilient design of the sensor network.  The designed estimator tolerates the removal/deletion of up to $q$ sensor nodes or communication links.
  
\subsection{Future Research}
The notion of redundant observers allows for the optimal design of sensor networks. In other words, one can optimally choose (in terms of cost and performance) among the observationally equivalent sensor measurements and redundant communication links. The proposed resilient design can be applied to distributed target tracking and intelligent transportation systems (vehicle platooning) as future research directions.

\bibliographystyle{IEEEbib}
\bibliography{bibliography}

\begin{thebibliography}{10}

\bibitem{sundaram_2021resilient}
A.~Mitra, F.~Ghawash, S.~Sundaram, and W.~Abbas,
\newblock ``On the impacts of redundancy, diversity, and trust in resilient
  distributed state estimation,''
\newblock {\em IEEE Transactions on Control of Network Systems}, 2021.

\bibitem{PIRANI2023111264}
M.~Pirani, A.~Mitra, and S.~Sundaram,
\newblock ``Graph-theoretic approaches for analyzing the resilience of
  distributed control systems: A tutorial and survey,''
\newblock {\em Automatica}, vol. 157, pp. 111264, 2023.

\bibitem{battilotti2018distributed}
S.~Battilotti, F.~Cacace, M.~d’Angelo, and A.~Germani,
\newblock ``Distributed kalman filtering over sensor networks with unknown
  random link failures,''
\newblock {\em IEEE Control Systems Letters}, vol. 2, no. 4, pp. 587--592,
  2018.

\bibitem{mohammadi2015distributed}
A.~Mohammadi and A.~Asif,
\newblock ``Distributed consensus $+ $ innovation particle filtering for
  bearing/range tracking with communication constraints,''
\newblock {\em IEEE Transactions on Signal Processing}, vol. 63, no. 3, pp.
  620--635, 2015.

\bibitem{tase}
M.~Doostmohammadian, A.~Taghieh, and H.~Zarrabi,
\newblock ``Distributed estimation approach for tracking a mobile target via
  formation of {UAVs},''
\newblock {\em IEEE Transactions on Automation Science and Engineering}, vol.
  19, no. 4, pp. 3765--3776, 2021.

\bibitem{ennasr2016distributed}
O.~Ennasr, G.~Xing, and X.~Tan,
\newblock ``Distributed time-difference-of-arrival {(TDOA)}-based localization
  of a moving target,''
\newblock in {\em IEEE 55th Conference on Decision and Control}. IEEE, 2016,
  pp. 2652--2658.

\bibitem{ennasr2018distributed}
O.~Ennasr and X.~Tan,
\newblock ``Distributed localization of a moving target: Structural
  observability-based convergence analysis,''
\newblock in {\em American Control Conference}. IEEE, 2018, pp. 2897--2903.

\bibitem{teixeira2014distributed}
A.~Teixeira, I.~Shames, H.~Sandberg, and K.~H. Johansson,
\newblock ``Distributed fault detection and isolation resilient to network
  model uncertainties,''
\newblock {\em IEEE transactions on Cybernetics}, vol. 44, no. 11, pp.
  2024--2037, 2014.

\bibitem{tcns20}
M.~Doostmohammadian and N.~Meskin,
\newblock ``Sensor fault detection and isolation via networked estimation:
  Full-rank dynamical systems,''
\newblock {\em IEEE Transactions on Control of Network Systems}, vol. 8, no. 2,
  pp. 987--996, 2021.

\bibitem{hajshirmohamadi2019event}
S.~Hajshirmohamadi, F.~Sheikholeslam, M.~Davoodi, and N.~Meskin,
\newblock ``Event-triggered simultaneous fault detection and tracking control
  for multi-agent systems,''
\newblock {\em International Journal of Control}, vol. 92, no. 8, pp.
  1928--1944, 2019.

\bibitem{davoodi2013distributed}
M.~Davoodi, K.~Khorasani, H.~A. Talebi, and H.~R. Momeni,
\newblock ``Distributed fault detection and isolation filter design for a
  network of heterogeneous multiagent systems,''
\newblock {\em IEEE Transactions on Control Systems Technology}, vol. 22, no.
  3, pp. 1061--1069, 2013.

\bibitem{ferrari2011distributed}
R.~M.~G. Ferrari, T.~Parisini, and M.~M. Polycarpou,
\newblock ``Distributed fault detection and isolation of large-scale
  discrete-time nonlinear systems: An adaptive approximation approach,''
\newblock {\em IEEE Transactions on Automatic Control}, vol. 57, no. 2, pp.
  275--290, 2011.

\bibitem{deng2023distributed}
Y.~Deng, Z.~Huang, Y.~Jia, Y.~Xu, and P.~Shi,
\newblock ``Distributed estimation and smoothing for linear dynamic systems
  over sensor networks,''
\newblock {\em IEEE Transactions on Circuits and Systems I: Regular Papers},
  2023.

\bibitem{wu2021secure}
H.~Wu, B.~Zhou, and C.~Zhang,
\newblock ``Secure distributed estimation against data integrity attacks in
  internet-of-things systems,''
\newblock {\em IEEE Transactions on Automation Science and Engineering}, vol.
  19, no. 3, pp. 2552--2565, 2021.

\bibitem{jin2023new}
L.~Jin, J.~Zhang, Y.~Qi, and S.~Li,
\newblock ``New distributed consensus schemes with time delays and output
  saturation,''
\newblock {\em IEEE Transactions on Automation Science and Engineering}, 2023.

\bibitem{liu2023fully}
D.~Liu, D.~Ye, and X.~Zhao,
\newblock ``Fully distributed secure state estimation for nonlinear multi-agent
  systems against dos attacks: An edge-pinning-based method,''
\newblock {\em IEEE Transactions on Automation Science and Engineering}, 2023.

\bibitem{zhu2022optimal}
M.~Zhu, R.~Wang, and X.~Sun,
\newblock ``The optimal distributed weighted least-squares estimation in finite
  steps for networked systems,''
\newblock {\em IEEE Transactions on Circuits and Systems II: Express Briefs},
  vol. 70, no. 3, pp. 1069--1073, 2022.

\bibitem{kar2013consensus}
S.~Kar and J.~M.~F. Moura,
\newblock ``Consensus+ innovations distributed inference over networks:
  cooperation and sensing in networked systems,''
\newblock {\em IEEE Signal Processing Magazine}, vol. 30, no. 3, pp. 99--109,
  2013.

\bibitem{das2015distributed}
S.~Das and J.~M.~F. Moura,
\newblock ``Distributed kalman filtering with dynamic observations consensus,''
\newblock {\em IEEE Transactions on Signal Processing}, vol. 63, no. 17, pp.
  4458--4473, 2015.

\bibitem{Silm2020A}
H.~Silm, D.~Efimov, W.~Michiels, R.~Ushirobira, and J.~Richard,
\newblock ``A simple finite-time distributed observer design for linear
  time-invariant systems,''
\newblock {\em Systems and Control Lettters}, vol. 141, pp. 104707, 2020.

\bibitem{khan2010connectivity}
U.~A. Khan, S.~Kar, A.~Jadbabaie, and J.~M.~F. Moura,
\newblock ``On connectivity, observability, and stability in distributed
  estimation,''
\newblock in {\em 49th IEEE Conference on Decision and Control}. IEEE, 2010,
  pp. 6639--6644.

\bibitem{chen2018internet}
Y.~Chen, S.~Kar, and J.~M.~F. Moura,
\newblock ``The internet of things: Secure distributed inference,''
\newblock {\em IEEE Signal Processing Magazine}, vol. 35, no. 5, pp. 64--75,
  2018.

\bibitem{acc13}
M.~Doostmohammadian and U.~A. Khan,
\newblock ``Topology design in networked estimation: a generic approach,''
\newblock in {\em American Control Conference}, Washington, DC, Jun. 2013, pp.
  4140--4145.

\bibitem{he2020secure}
X.~He, X.~Ren, H.~Sandberg, and K.~H Johansson,
\newblock ``How to secure distributed filters under sensor attacks,''
\newblock {\em IEEE Transactions on Automatic Control}, vol. 67, no. 6, pp.
  2843--2856, 2021.

\bibitem{battilotti2021stability}
S.~Battilotti, F.~Cacace, and M.~d’Angelo,
\newblock ``A stability with optimality analysis of consensus-based distributed
  filters for discrete-time linear systems,''
\newblock {\em Automatica}, vol. 129, pp. 109589, 2021.

\bibitem{olfati2007distributed}
R.~Olfati-Saber,
\newblock ``Distributed kalman filtering for sensor networks,''
\newblock in {\em 46th IEEE Conference on Decision and Control}. IEEE, 2007,
  pp. 5492--5498.

\bibitem{galvez2021cycle}
W.~G{\'a}lvez, F.~Grandoni, A.~Jabal~Ameli, and K.~Sornat,
\newblock ``On the cycle augmentation problem: hardness and approximation
  algorithms,''
\newblock {\em Theory of Computing Systems}, vol. 65, pp. 985--1008, 2021.

\bibitem{byrka2020breaching}
J.~Byrka, F.~Grandoni, and A.~Jabal~Ameli,
\newblock ``Breaching the 2-approximation barrier for connectivity
  augmentation: a reduction to steiner tree,''
\newblock in {\em 52nd Annual ACM SIGACT Symposium on Theory of Computing},
  2020, pp. 815--825.

\bibitem{cecchetto2021bridging}
F.~Cecchetto, V.~Traub, and R.~Zenklusen,
\newblock ``Bridging the gap between tree and connectivity augmentation:
  unified and stronger approaches,''
\newblock in {\em 53rd Annual ACM SIGACT Symposium on Theory of Computing},
  2021, pp. 370--383.

\bibitem{galvez2021approximation}
W.~G{\'a}lvez, F.~Sanhueza-Matamala, and J.~A. Soto,
\newblock ``Approximation algorithms for vertex-connectivity augmentation on
  the cycle,''
\newblock in {\em International Workshop on Approximation and Online
  Algorithms}. Springer, 2021, pp. 1--22.

\bibitem{gupta2011approximation}
A.~Gupta and J.~K{\"o}nemann,
\newblock ``Approximation algorithms for network design: A survey,''
\newblock {\em Surveys in Operations Research and Management Science}, vol. 16,
  no. 1, pp. 3--20, 2011.

\bibitem{icassp2016}
M.~Doostmohammadian and U.~A. Khan,
\newblock ``Measurement partitioning and observational equivalence in state
  estimation,''
\newblock in {\em IEEE Conference on Acoustics, Speech and Signal Processing},
  2016, pp. 4855--4859.

\bibitem{acc13_kar}
S.~Pequito, S.~Kar, and A.~P. Aguiar,
\newblock ``A structured systems approach for optimal actuator-sensor placement
  in linear time-invariant systems,''
\newblock in {\em American Control Conference}, 2013, pp. 6123--6128.

\bibitem{liu-pnas}
Y.~Y. Liu, J.~J. Slotine, and A.~L. Barab\'{a}si,
\newblock ``Observability of complex systems,''
\newblock {\em Proceedings of the National Academy of Sciences}, vol. 110, no.
  7, pp. 2460--2465, 2013.

\bibitem{godsil}
C.~Godsil and G.~Royle,
\newblock {\em Algebraic graph theory},
\newblock New York: Springer, 2001.

\bibitem{chen2022survivable}
H.~Chen, X.~Wang, Z.~Liu, Z.~Li, and L.~Shen,
\newblock ``Survivable networks for consensus,''
\newblock {\em IEEE Transactions on Control of Network Systems}, vol. 9, no. 2,
  pp. 588--600, 2022.

\bibitem{lau2007survivable}
L.~Lau, J.~Naor, M.~R. Salavatipour, and M.~Singh,
\newblock ``Survivable network design with degree or order constraints,''
\newblock in {\em Proceedings of the 39th annual ACM symposium on theory of
  computing}, 2007, pp. 651--660.

\bibitem{tnse18}
M.~Doostmohammadian, H.~R. Rabiee, H.~Zarrabi, and U.~Khan,
\newblock ``Observational equivalence in system estimation: Contractions in
  complex networks,''
\newblock {\em IEEE Transactions on Network Science and Engineering}, vol. 5,
  no. 3, pp. 212--224, 2018.

\bibitem{pearce2016space}
D.~J. Pearce,
\newblock ``A space-efficient algorithm for finding strongly connected
  components,''
\newblock {\em Information Processing Letters}, vol. 116, no. 1, pp. 47--52,
  2016.

\bibitem{hagerup2020space}
T.~Hagerup,
\newblock ``Space-efficient dfs and applications to connectivity problems:
  simpler, leaner, faster,''
\newblock {\em Algorithmica}, vol. 82, no. 4, pp. 1033--1056, 2020.

\bibitem{algorithm}
T.~H. Cormen, C.~E. Leiserson, R.~L. Rivest, and C.~Stein,
\newblock {\em {Introduction to Algorithms}},
\newblock MIT Press, 2009.

\bibitem{xiao2007distributed}
L.~Xiao, S.~Boyd, and S.~Kim,
\newblock ``Distributed average consensus with least-mean-square deviation,''
\newblock {\em Journal of parallel and distributed computing}, vol. 67, no. 1,
  pp. 33--46, 2007.

\bibitem{bay}
J.~Bay,
\newblock {\em Fundamentals of linear state space systems},
\newblock McGraw-Hill, 1999.

\bibitem{tsipn}
M.~Doostmohammadian and U.~A. Khan,
\newblock ``Minimal sufficient conditions for structural
  observability/controllability of composite networks via kronecker product,''
\newblock {\em IEEE Transactions on Signal and Information processing over
  Networks}, vol. 6, pp. 78--87, 2019.

\bibitem{el1997cone}
L.~El~Ghaoui, F.~Oustry, and M.~AitRami,
\newblock ``A cone complementarity linearization algorithm for static
  output-feedback and related problems,''
\newblock {\em IEEE Transactions on Automatic Control}, vol. 42, no. 8, pp.
  1171--1176, 1997.

\bibitem{usman_cdc:11}
U.~A. Khan and A.~Jadbabaie,
\newblock ``Coordinated networked estimation strategies using structured
  systems theory,''
\newblock in {\em 49th IEEE Conference on Decision and Control}, Orlando, FL,
  Dec. 2011, pp. 2112--2117.

\bibitem{mangasarian1995extended}
O.~Mangasarian and J.~Pang,
\newblock ``The extended linear complementarity problem,''
\newblock {\em SIAM Journal on Matrix Analysis and Applications}, vol. 16, no.
  2, pp. 359--368, 1995.

\bibitem{pajic2010wireless}
M.~Pajic, S.~Sundaram, J.~Le~Ny, G.~J. Pappas, and R.~Mangharam,
\newblock ``The wireless control network: Synthesis and robustness,''
\newblock in {\em 49th IEEE Conference on Decision and Control}. IEEE, 2010,
  pp. 7576--7581.

\bibitem{ferentinos2007adaptive}
K.~P. Ferentinos and T.~A. Tsiligiridis,
\newblock ``Adaptive design optimization of wireless sensor networks using
  genetic algorithms,''
\newblock {\em Computer Networks}, vol. 51, no. 4, pp. 1031--1051, 2007.

\bibitem{xu2005optimal}
K.~Xu, Q.~Wang, H.~Hassanein, and G.~Takahara,
\newblock ``Optimal wireless sensor networks {(WSNs)} deployment: minimum cost
  with lifetime constraint,''
\newblock in {\em IEEE International Conference on Wireless And Mobile
  Computing, Networking And Communications}. IEEE, 2005, vol.~3, pp. 454--461.

\bibitem{aysal2008distributed}
T.~C. Aysal, M.~J. Coates, and M.~G. Rabbat,
\newblock ``Distributed average consensus with dithered quantization,''
\newblock {\em IEEE transactions on Signal Processing}, vol. 56, no. 10, pp.
  4905--4918, 2008.

\bibitem{erseghe2011fast}
T.~Erseghe, D.~Zennaro, E.~Dall'Anese, and L.~Vangelista,
\newblock ``Fast consensus by the alternating direction multipliers method,''
\newblock {\em IEEE Transactions on Signal Processing}, vol. 59, no. 11, pp.
  5523--5537, 2011.

\bibitem{Ding2022Resilient}
D.~Derui, L.~Huanyi, D.~Hongli, and L.~Hongjian,
\newblock ``Resilient filtering of nonlinear complex dynamical networks under
  randomly occurring faults and hybrid cyber-attacks,''
\newblock {\em IEEE Transactions on Network Science and Engineering}, vol. 9,
  pp. 2341--2352, 2022.

\bibitem{Qnode}
V.~K. Garg,
\newblock ``Q-node connectivity in graphs,''
\newblock {\em Journal of Algorithms}, vol. 13, no. 2, pp. 298--309, 1992.

\bibitem{cheriyan2001rooted}
J.~Cheriyan, T.~Jord{\'a}n, and Z.~Nutov,
\newblock ``On rooted node-connectivity problems,''
\newblock {\em Algorithmica}, vol. 30, pp. 353--375, 2001.

\bibitem{Kang1991Hamiltonian}
Feng K. and Qin M.,
\newblock ``Hamiltonian algorithms for hamiltonian systems and a comparative
  numerical study,''
\newblock {\em Computer Physics Communications}, vol. 65, pp. 173--187, 1991.

\bibitem{Boyd2022Approximation}
S.~P. Boyd, J.~Cheriyan, A.~Haddadan, and S.~Ibrahimpur,
\newblock ``Approximation algorithms for flexible graph connectivity,''
\newblock {\em Mathematical Programming}, pp. 1--24, 2022.

\end{thebibliography}

\begin{IEEEbiography}[{\includegraphics[width=1.1in,clip,keepaspectratio]{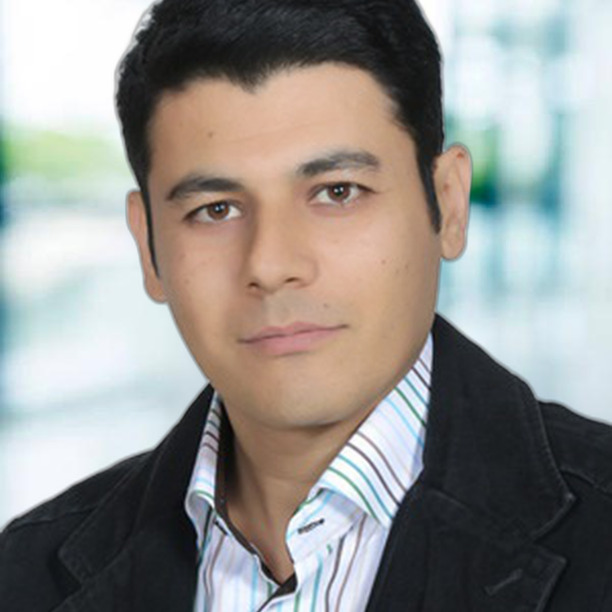}}]{Mohammadreza~Doostmohammadian}
	received his B.Sc. and M.Sc. in Mechanical Engineering from Sharif University of Technology, respectively in 2007 and 2010, where he worked on applications of control systems and robotics. He received his PhD in Electrical and Computer Engineering from Tufts University, MA, USA in 2015. During his PhD in Signal Processing and Robotic Networks (SPARTN) lab, he worked on control and signal processing over networks with applications in social networks. From 2015 to 2017 he was a postdoc researcher at ICT Innovation Center for Advanced Information and Communication Technology (AICT), School of Computer Engineering, Sharif University of Technology, with research on network epidemic, distributed algorithms, and complexity analysis of distributed estimation methods. He was a researcher at Iran Telecommunication Research Center (ITRC), Tehran, Iran in 2017 working on distributed control algorithms and estimation over IoT. Since 2017 he has been an Assistant Professor with the Mechatronics Department at Semnan University,  Iran. He was a visiting researcher at the School of Electrical Engineering and Automation, Aalto University, Espoo, Finland, working on constrained and unconstrained distributed optimization techniques and their applications. His general research interests include distributed estimation, control, learning, and optimization over networks. He was the chair of the robotics and control session at the ISME-2018 conference, the session chair at the 1st Artificial Intelligent Systems Conference of Iran, 2022, and the IPC member and Associate Editor of IEEE/IFAC CoDIT2024 conference.
\end{IEEEbiography}

\begin{IEEEbiography}[{\includegraphics[width=1.1in,clip,keepaspectratio]{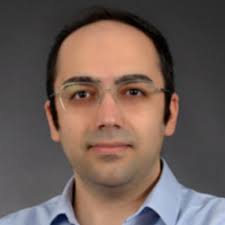}}]{Mohammad~Pirani}
	received the B.Sc. degree in mechanical engineering from the Amirkabir University of Technology in 2011, the M.Sc. degree in electrical and computer engineering, and the Ph.D. degree in mechanical and mechatronics engineering, both from the University of Waterloo in 2014 and 2017, respectively. He is currently an Assistant Professor at the Department of Mechanical Engineering, University of Ottawa, Canada. Before this, he was a Postdoctoral Researcher with the Department of Electrical and Computer Engineering, the University of Toronto. From 2018 to 2019, he was a Postdoctoral Researcher with the Department of Automatic Control, KTH Royal Institute of Technology, Sweden.  His research interests include resilient and fault-tolerant control, networked control systems, and multi-agent systems.
\end{IEEEbiography}
\end{document}